\documentclass{llncs}

\usepackage{lmodern}

\usepackage[T1]{fontenc}
\usepackage{amsmath}
\usepackage{amssymb}
\usepackage[utf8]{inputenc}
\usepackage{graphicx}
\usepackage{microtype}

\newcommand{\junk}[1]{}

\usepackage{algorithm}

\usepackage{algorithmic}

\title{Computing the Boolean product
of two $n\times n$ Boolean matrices using
$O(n^2)$ mechanical operations}

\date{}

\author{Andrzej Lingas \inst{1}
\and Mia Persson \inst{2}}

\institute{Department of Computer Science, Lund University\\
\email{Andrzej.Lingas@cs.lth.se}
\and Department of Computer Science, Malm\"o University\\
\email{Mia.Persson@mau.se}}

\begin{document}
\maketitle

\begin{abstract}
We study the problem of determining
the Boolean product
of two $n\times n$ Boolean matrices
in an unconventional computational
model allowing for mechanical operations.
We show that $O(n^2)$ operations are sufficient
to compute the product in this model.
\end{abstract}

\section{Introduction}

The problem of Boolean matrix multiplication has a wide range of
fundamental applications, for instance, in graph algorithms (see,
e.g., \cite{Z02}).  The Boolean product of two Boolean $n\times n$
matrices can be computed using $O(n^3)$ Boolean AND and OR binary
operations following its definition. This is optimal if only these two
types of operations are allowed \cite{MG76,Pa75,P75}.  If also the one argument NOT
operation is allowed then the product can be computed using
$O(n^{\omega})$ operations, where $\omega$ stands for the exponent of
the so called fast arithmetic matrix multiplication.  The
$O(n^{\omega})$ bound follows by a straightforward reduction of
Boolean matrix multiplication to the arithmetic one for $0-1$ matrices.
Already five decades ago, 
Strassen presented an algorithm for arithmetic matrix multiplication
breaking out of the standard method and
using only $O(n^{2.8074})$ multiplications, additions and subtractions \cite{S69}.
The following series of improvements of the exponent
$\omega$ of fast matrix multiplication culminates
in the recent results of
Le Gall and Vassilevska Williams showing
$\omega < 2.373$
\cite{LG14,Vassilevska12}. On the other hand, Alman
et al. have recently shown that there is an $\ell > 2$
such that one cannot derive
an upper bound on $\omega$ better than $\ell$
by using the known approaches \cite{AV18}.
Also, the substantially sub-cubic algebraic algorithms
for matrix multiplication have very large overhead
which makes them hardly practical. Unfortunately,
no the so called combinatorial (i.e., with small constants)
algorithm for Boolean matrix multiplication
running in $O(n^{3-\epsilon})$ time, for a fixed $\epsilon > 0,$
is known \cite{BW12,VWW}.
The fastest known combinatorial algorithm
for Boolean matrix multiplication due to Yu runs in 
$O(n^3 poly(\log \log n)/ \log^4n))$ time \cite{Yu15}.
Hence, a natural question arises about the complexity of
Boolean matrix multiplication in different models 
of computation, possibly more powerful or unconventional.

To begin with, if one uses huge integers, requiring
$O(n\log n)$ bits then the Boolean product of two $n\times n$ Boolean
matrices can computed in $O(n^2)$ steps on RAM as observed in \cite{PS16}.
Recently, also a nondeterministic algorithm for $n\times n$ matrix
multiplication using $O(n^2)$ arithmetic operations has been presented
by Korec and Wiedermann in \cite{KW14}.  It results from a
derandomization of Freivalds' $O(n^2)$-time randomized algorithm for
integer matrix product verification \cite{F77}. Simply, the algorithm guesses
first the product matrix and then verifies its correctness. Again, the
derandomization involves huge numbers requiring $O(n)$ times more bits
than the input numbers \cite{KW14}. More recently, Wiedermann has
presented two further, slightly slower nondeterministic algorithms for
matrix multiplication, both running in $O(n^2\log n)$ time and relying
on the derandomization of Freivalds' algorithm. The first runs on a
real RAM, the second on a unit-cost RAM using only integers of size
proportional to that of the largest entry in the input matrices
\cite{W14}. On the other hand, no quantum algorithm for Boolean matrix
product faster than those based on fast algebraic matrix
multiplication in the standard computational model has been 
devised so far \cite{{LeGallisaac12}}.

From the perspective of modern electronic computing, mechanical
computing seems not only old-fashioned but even unconventional. The
history of mechanical devices assisting computing stretches several
thousands years back, e.g., the so called Roman abacus was used as
early as 2700-2300 B.C. in Babylonia \cite{Wik}.  Under several
thousands of years mechanical devices assisted astronomical and
navigation calculations.  The slide-rule has been used since 1620 and
mechanical calculators were intensively developed in 17, 18 and 19
centuries (e.g., by Pascal, Leibniz and Babbage).

In this paper, we study the problem of determining the Boolean product
of two $n\times n$ Boolean matrices in an unconventional computational
model allowing for mechanical operations.  We show that the Boolean
product of an $n\times n$ Boolean matrix with an $n$-dimensional
Boolean column vector can be computed using $O(n)$ operations.  Hence,
we can infer that $O(n^2)$ operations are sufficient to compute the
Boolean product of two $n\times n$ Boolean matrices in this model.
For smaller matrices the mechanical operations can be performed even
manually while for the larger ones the operations can be automatized
using electric power.  Our result demonstrates that already ancient
civilizations had sufficient technology to perform relatively fast
matrix multiplication of moderately large Boolean matrices.
To the best of our knowledge, we are not familiar with any
prior mechanical approach to (Boolean) matrix multiplication.

\subsection{Basic ideas}

Our mechanical algorithm for Boolean matrix multiplication computes
the Boolean product of the two input $n\times n$ Boolean matrices by
repetitive multiplication of the first matrix with consecutive columns
of the second matrix.  While computing the Boolean product of a fixed
Boolean $n\times n$ matrix with a Boolean column vector one encounters
two major bottlenecks that prohibit obtaining a substantially
sub-quadratic in $n$ number of operations in the standard
computational model (solely poly-logarithmic speedups are known after
preprocessing \cite{RW}).  Note that only the $j$-th columns of the
fixed matrix where the $j$-th coordinate of the input column vector is
$1$ can affect the output product vector. The first bottleneck is the
selection of the aforementioned columns (alternatively, activating
them and deactivating the remaining ones). The other bottleneck is the
computation of the disjunction of entries belonging to active columns
for each row of the fixed matrix. We design a special purpose
mechanical processor for computing the Boolean product of a Boolean
$n\times n$ matrix with a Boolean $n$-dimensional column vector that
solves each of the two bottleneck problems using $O(n)$ mechanical
operations.
\junk{\footnote{Any substantially subquadratic 
in $n$ solution to the bottleneck problems
(even allowing for a subcubic in $n$ preprocessing of the matrix)
in a standard sequential model
would lead to substantially subcubic in $n$ combinatorial
algorithm for Boolean matrix multiplication,  which would
be a breakthrough.}.}
The matrix-vector processor (MVP) needs also to perform a
number of other operations, e.g., reading the input vector or the
input matrix or outputting the product vector etc., that can be
implemented purely mechanically (if one insists on that) in many ways.
We leave here the details to the reader.

\section{Requirements on a processor for matrix-vector product}
We shall use a special purpose processor for computing
the Boolean product of a Boolean $n\times n$ matrix with a Boolean
$n$-dimensional column vector. We shall term such a processor MVP for short. In
particular, an MVP has an input $n \times n$ array part, an input
$n$-dimensional vector part, and an output $n$-dimensional vector
part, see Fig. 3.
In this section, we shall provide solely general requirements
that an MVP should satisfy.  In Section 4, we shall outline
implementations of MVP including mechanical operations that allow for
fulfillment of these requirements. The list of requirements is as follows.

\begin{enumerate}
\item A Boolean $n\times n$ matrix can be read into the MVP input
array using $O(n^2)$ operations.
\item A Boolean $n$-dimensional vector can be read into
the MVP input vector using $O(n)$ operations.
\item The MVP output vector can be reported using $O(n)$ operations.
\item All $j$-th columns of the MVP input array can be
appropriately activated/deactivated (if necessary) such that afterwards
the $j$-th column is active if and only if the $j$-th
coordinate of the MVP input vector is set to $1$,
using totally $O(n)$ operations.
\item For all $i$-th rows of the MVP input array that 
include at least one $1$ belonging to an active
column of the array the $i$-th coordinate
of the MVP output vector can be set to $1$ while
all other coordinates of the output vector
can be set to $0$  using
totally $O(n)$ operations.
\item The MVP output
vector can be reset to an initial state using
$O(n)$ operations.
\end{enumerate}

\section{The algorithms for Boolean matrix-vector and matrix-matrix products}

In this section, we present two simple mechanical algorithms for
Boolean matrix-vector product and Boolean matrix-matrix product,
respectively.  Both are based on the use of an MVP. It will be
convenient to refer to the first algorithm as a functional procedure
$Matvec(A,V)$ in order to facilitate its repetitive use in the second
algorithm.

\begin{figure}[t]
\centering
\fbox{\begin{minipage}[t]{0.90\textwidth}
\begin{algorithmic}[1]
\REQUIRE a Boolean $n\times n$ matrix $A$ and
a Boolean $n$-dimensional column vector $V$ 
in the MVP input array and the MVP input vector,
respectively.
\ENSURE the Boolean column vector product of
$A$ and $V.$
\FOR {$j = 1$ {\bf to} $n$}
\STATE {\bf if} $V[j]=1$ and the $j$-th
column of the MVP input array is not active {\bf then} activate the $j$-th
column
\STATE {\bf if} $V[j]=0$ and the $j$-th
column of the MVP input array is active {\bf then} deactivate the $j$-th
column
\ENDFOR
\FOR {$i = 1$ {\bf to} $n$}
\STATE {\bf if} the $i$-th row of the MVP input array includes
at least one $1$ that belongs to an active
column of the array {\bf then} set the $i$-th
coordinate of the MVP output vector to $1$
{\bf else} set this coordinate of the MVP output vector to $0$
\ENDFOR
\RETURN the output vector
\end{algorithmic}
\end{minipage}}
\caption{The functional procedure Matvec(A,V)}
\label{fig: algo1}
\end{figure}

\begin{lemma}\label{lem: matvec}
The procedure Matvec(A,V) computes the
Boolean column vector product of $A$
and $V$ correctly, using $O(n)$ operations.
\end{lemma}
\begin{proof}
The correctness of the procedure follows from
the following two facts. 
\begin{enumerate}
\item
The $i$-th coordinate of the MVP output vector is set
to $1$ if and only if the $i$-th row of
the MVP input array includes at least one $1$
belonging to an active column of the array.
\item The $j$-th
column of the MVP input array is active if and only
if the $j$-th coordinate of the MVP input vector
is set to $1.$
\end{enumerate}
The upper bound on the number of operations
necessary to implement $Matvec(A,V)$   follows from
the requirements on an MVP (see the previous section).
In particular the upper bound on the total
number of operations necessary to implement 
the first loop  follows from Requirement 4 while that for
the second loop from Requirement 5. 
\junk{
If columns of the input array are activated separately
than Step 1 requires at most $n$ activation operations
and $O(n)$ walking steps to scan the input vector
and activate the appropriate columns. 
If we use a more refined mechanical processor then
the activation of the columns of the input array 
corresponding to $1$ in the input vector 
be even done just by lowering the representation
of the input vector appropriately using $O(1)$
operations. 
Step 2 requires $n$ ladder movements and 
$O(n)$ walking steps.}
\qed
\end{proof}

By repetitively applying $Matvec(A,B[*,j])$
to consecutive columns $B[*,j])$ of the
second input Boolean matrix $B,$ we obtain our
algorithm (Algorithm 1) for Boolean matrix-matrix product.

\begin{figure}[t]
\centering
\fbox{\begin{minipage}[t]{0.90\textwidth}
\begin{algorithmic}[1]
\REQUIRE Boolean $n\times n$ matrices $A$ and
$B.$
\ENSURE the Boolean matrix product of
$A$ and $B$
\STATE read the matrix $A$ into the MVP input array
\FOR {$j = 1$ {\bf to}  $n$}
\STATE read the $j$-th column $B[*,j]$ of $B$ in the MVP input vector
\RETURN $Matvec(A,B[*,j])$
\STATE reset the MVP output vector to their initial state
\ENDFOR
\end{algorithmic}
\end{minipage}}
\caption{The mechanical algorithm for Boolean matrix multiplication (Algorithm 1)}
\label{fig: algo1}
\end{figure}

\begin{theorem}
Algorithm 1 computes the Boolean product of the
matrices $A$ and $B$ correctly, using
$O(n^2)$ operations.
\end{theorem}
\begin{proof}
The correctness of the algorithm follows from the fact
that it computes the consecutive columns of
the Boolean product of $A$ and $B$ correctly
by
\newline
 Lemma \ref{lem: matvec}.
The $n$ calls of the procedure $Matvec$ require
totally $O(n^2)$ operations by Lemma \ref{lem: matvec}.
Reading the matrix $A$ in the input array
requires $O(n^2)$ operations by Requirement 1 on an MVP.
Reading the consecutive columns of the matrix $B$
into the MVP input vector 
totally requires $O(n^2)$ operations by Requirement 2.
Finally, resetting the MVP output vectors requires
$O(n)$ operations by Requirement 6.
\qed
\end{proof}

\begin{remark}
Observe that all the requirements on MVP but 
for the bottleneck ones 4 and 5 can be easily
realized up to logarithmic factors
in such standard sequential computational models
as RAM, Turing machine or Boolean circuits. Hence,
if the bottleneck requirements could be
realized in substantially subquadratic 
in $n$ time
(even allowing for a subcubic in $n$ preprocessing of the matrix)
by combinatorial algorithms in the aforementioned standard models, it 
would yield substantially subcubic in $n$ combinatorial
algorithms for Boolean matrix multiplication,  which would
be a breakthrough (cf. \cite{Yu15,RW}).
\end{remark}

\section{A mechanical MVP}

In this section, we outline an idea of implementation of an MVP with
help of mechanical operations.

Each column of the MVP input array is modelled by an axis with
attached equal size segments representing consecutive column
entries. Each such a segment is either aligned to the axis which is
interpreted as the corresponding entry is set to $0$ or it is set into
a perpendicular to the axis position which is interpreted as the
corresponding entry is set to $1.$ See Fig. 4
for an example.  If the column
modelled by the axis is not active,
the perpendicular segments are placed horizontally. If the column
should be activated then the 
axis rotates $90$ degrees 
so all
the attached perpendicular segments take vertical positions under the axis. 
Symmetrically, if the column is active and it should be deactivated
the axis modelling it is rotated in the reverse direction
by $90$ degrees. Such an
activation/deactivation  of a column is counted as a single operation. There are many
possible ways to implement the rotation of the axis by $90$ degrees
that we omit here. The whole process of activating/deactivating (if necessary)
the $j$-th columns of the
input array such that 
afterwards the $j$-th column is active
if and only if the $j$-th coordinate of the MVP input vector is
set to $1$ 
requires at most $n$ operations of column activations/deactivations  and
scanning the representation of the input MVP vector. For instance, we
may assume that the representation of the input vector is placed
perpendicularly to the column axes close to their start-points so the
whole process of activating/deactivating columns requires $O(n)$
walking and scanning steps besides the $O(n)$ activation/deactivation
operations. In this way, Requirement 4 can be fulfilled.

In order to fulfill Requirement 5 for each row of the MVP input array,
we place a ladder under the segments attached to the axes
corresponding to the row, in the initial state when no column is
active.  The ladders are placed perpendicularly to the axes,
 see Fig. 5, 6.
They can
have some wheels to facilitate their horizontal  movements in the
direction perpendicular to the axes.  
We shall term  the rectangular space between two successive ladder steps
and the arms of the ladder as an opening. Each opening of the ladder
lies exactly under one of the segments attached
to an axis in the initial state when no
column is active. Importantly, the placements and sizes of the axes,
ladders, segments etc.  are so chosen that the rotations of the axes
with attached perpendicular segments are not obstructed by the
ladders.  If a segment is placed vertically after the rotation then it
goes through and a bit sticks out under the corresponding ladder
opening, see Fig. 5.
A representation of the MVP output vector is placed along
the rightmost axis. Initially, all coordinates of the output vector are
set to $1.$ In order to set the $i$-th coordinate of the output
vector, the ladder corresponding to the $i$-th row of the input array
is moved from the left to the right at most by the length of one
opening, see Fig. 6.
If no perpendicular segment sticks through any of its
openings such a movement by the length of one opening is possible and
the ladder hits the representation of the output vector in the section
corresponding to its $i$-th coordinate switching its state from $1$ to
$0.$ (Again, there are many possible technical ways of representing
the MVP output vector and implementing such a switch that we omit
here.)  Otherwise, such a full movement is not possible and the state
of the section remains set to $1.$ In effect, the $i$-th coordinate of
the output vector is set to $1$ if and only if the $i$-th row of the
input array includes at least one $1$ belonging to an active column of
the array.  Each such a movement of the ladder with a possible triggering
the switch of the state of the corresponding section is counted as ,say,
$O(1)$ operations. Hence, setting all the coordinates of the output
vector requires $O(n)$ operations and Requirement 5 is fulfilled.

To read a Boolean $n\times n$ matrix $A$ into the MVP input array
one can traverse the array for example in a column-wise
snake fashion and pull out segments corresponding to
the $1$ entries of the matrix $A$ that are
aligned to their axis and push back the pulled out segments 
that correspond to the $0$ entries of the matrix $A.$ 
Thus, we may assume that the reading of a Boolean $n\times n$
matrix in the input array 
requires $O(n^2)$ operations.

Similarly, reading a Boolean $n$-dimensional vector
into the MVP input vector as well as outputting
the MVP output vector
require $O(n)$ operations. 
Finally, to reset the MVP output vector, we need to
pull back the $n$ ladders to the left to their
initial positions and switch all the $0$ coordinates
of the output vector back to $1$ using $O(n)$
operations. It follows that the remaining requirements can be
fulfilled in this way.
\begin{figure}\label{fig: parts}
\begin{center}
\includegraphics[height=3.5cm]{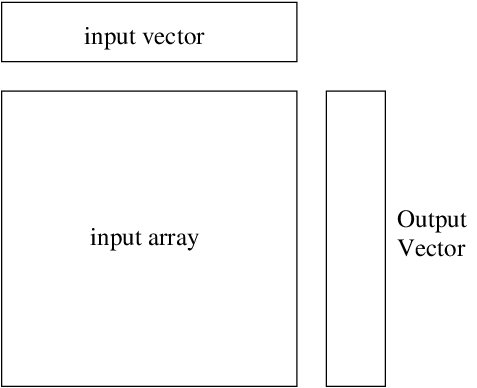}
\end{center}
\caption{The basic parts of the mechanical MVP}
\end{figure}

\begin{figure}\label{fig: axis}
\begin{center}
\includegraphics[height=4cm]{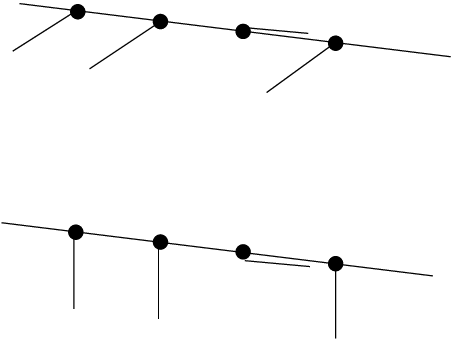}
\end{center}
\caption{(a) An axis modeling a not active column 
with entries $1,\ 1,\ 0, \ 1.$
(b) In order to activate the column the axis
rotates by 90 degrees.}
\end{figure}

\begin{figure}\label{fig: stics}
\begin{center}
\includegraphics[height=2cm]{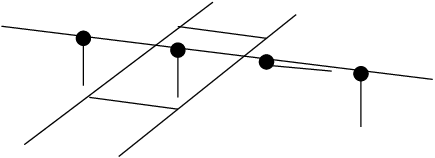}
\end{center}
\caption{A vertical segment attached to an axis
representing a $1$ entry of the active column
modelled by the axis sticks through an opening
of a ladder. The segment blocks the full movement
of the ladder by the length of one opening to the right.
In result, the coordinate of the MVP output vector
remains set to $1.$}
\end{figure}

 \begin{figure}\label{fig: array}
\begin{center}
\includegraphics[height=6cm]{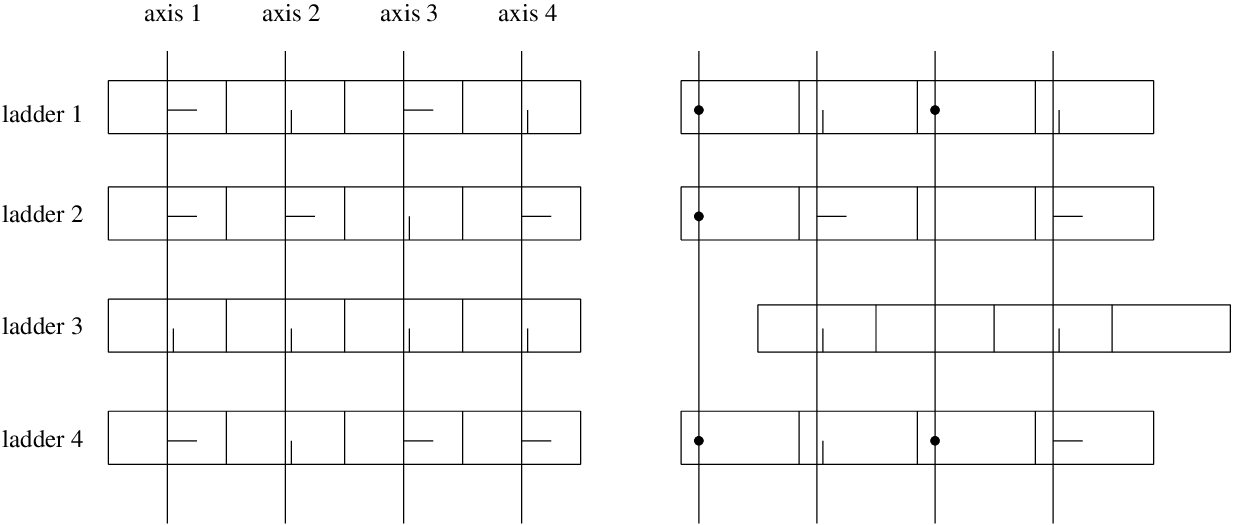}
\end{center}
\caption{(a) The MVP input array with 
columns set to $1101,$ $0100,$ $1001,$
and $0101,$ respectively. The columns
are not active and the ladders are in their
initial positions.
(b) The input array after the activation
of the first and third columns
and the movement of the ladders to the right. }
\end{figure}

\subsection{A parallel mechanical MVP}

In fact, the activations/deactivations of the appropriate columns of
the MVP input array as well as setting the values of the coordinates
of the MVP output vector can be done in parallel.

It is not difficult to design a bit more complex mechanical MVP where
an appropriate representation of the MVP input vector is moved towards
the start-points of the axes and the $1$s in the representation hit
the start-pints triggering the activation of the corresponding columns
of the input array.  More exactly, if the $j$-th coordinate of the
input vector is set to $1$ then the representation of this $1$ hits the
start-point of the axis modelling the $j$-th column casing its
activation.  Hence, the activations of all the columns can be done in
one step of the movement of the representation of the input vector.
In order to fulfill Requirement 4, we need to deactivate all columns
first before the activation movement.  The deactivation of all the
columns can be obtained by the reverse movement of the representation
of the previous input vector.

As for setting the output vector, i.e., Requirement 5, all the $n$
movements of the ladders to the right can be done
independently of each other in one parallel step.  (Still, if one
would like to count the ladder movements as a single operation, one
could elevate slightly the whole left part of the processor to cause
sliding the ladders to the right at most by one opening length etc.)  The
same holds for the reverse movements of the ladders back to the left,
i.e., resetting the output vector.  Finally, the reading of an input
vector into the MVP input vector can be also done in a single
parallel step using $O(n)$ operations
while the reading of the input matrix into the input array 
can be done in $O(n)$ parallel steps, each using $O(n)$ operations.

Hence, we obtain the following lemma and theorem.

\begin{lemma}
By using the parallel mechanical MVP,
the procedure $Matvec(A,V)$ can be implemented
in $O(1)$ parallel steps using $O(n)$ operations.
\end{lemma}

\begin{theorem}
By using the parallel mechanical MVP,
the Boolean product of two Boolean $n\times n$ 
matrices can be computed
in $O(n)$ parallel steps, each using $O(n)$ operations.
\end{theorem}

\section{An alternative implementation of MVP}

\begin{figure}\label{fig: wall}
\begin{center}
\includegraphics[height=6cm]{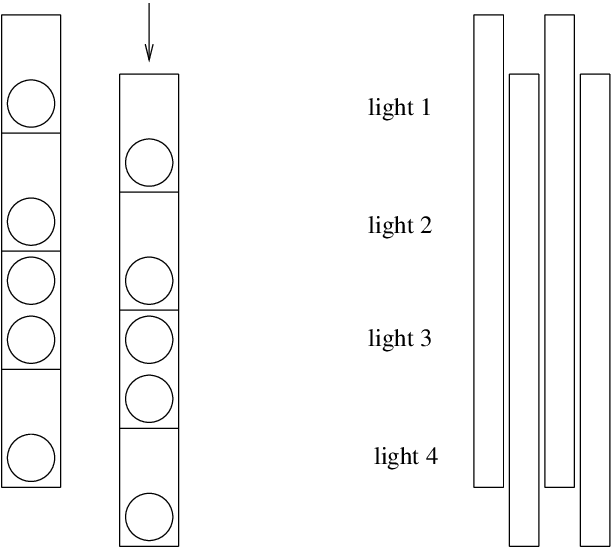}
\end{center}
\caption{(a) A thin wall modeling a not active column
of the MVP input array
with entries $1,\ 1,\ 0, \ 1.$
To activate the column the wall has to be moved down.
(b) The second and fourth columns of the MVP input array are active.}
\end{figure}

Here, we shall outline an alternative implementation
of  MVP using mechanical and light operations.
Each column of the MVP input $n\times n$ array is modelled
by a thin movable wall vertically divided into $2n$
sections. Each even numbered section has an always open
round window of the same size in its middle. The $2k-1$-st
section models the $k$-th entry
of the column in the input array. It has also
a window of the same shape and size in the middle.
The window can be either open or closed. The former is
interpreted as setting the $k$-th entry in the column
to $0$ while the latter as setting this entry to $1.$
Now, the activation of a column in the input array
consists in shifting the corresponding wall by
the length of a single section down.
See Fig. 7
for an example.
The deactivation of the column is obtained
by the reverse move of the wall up.
Both moves are counted as single  operations.
Hence, arguing similarly as in the previous section,
we infer that Requirement 4 can be fulfilled.

To fulfill Requirement 5, lights are placed in front of even sections
of the first wall in its initial deactivated position. 
See Fig. 7.
The distance
between two consecutive walls is very small.  Such an $i$-th light can
be observed on the other side at the last wall through the windows of
walls modelling deactivated columns and the open windows of walls
modelling active columns if and only if the $i$-th row of the input
array does not contain any $1$ belonging to an active column. Thus, in
case the light is observed the $i$-th coordinate of the MVP output
vector is set to $0$ and otherwise it is set to $1.$ Since the
observation of the light is counted as a single observation,
Requirement 5 can be also fulfilled. The fulfillment of the remaining
MVP requirements do not require any involved ideas and can be achieved
in many ways so we omit its discussion.

\section{Final remarks}
The 
mechanical operations used in our method 
like turning or shifting mechanical units
are very slow compared to the electronic ones.
Also, it would be technically hard to apply the aforementioned
 mechanical operations
to very large matrices. For these reasons, in the range of matrix
sizes feasible for the mechanical operations, Boolean matrix
multiplication can be performed much faster by electronic computers,
even if they run the straightforward cubic-time algorithm.'
Still, it seems of at least theoretical interest that in the realm
of mechanical computations one can design an algorithm for
Boolean matrix multiplication using a quadratic number of operations
in contrast to the straightforward one using a cubic number
of operations.  
The two ideas behind our mechanical algorithm for
Boolean matrix multiplication are switching on/off
all entries of an array column in a single mechanical step   
and computing the disjunction of switched on entries
of an array row in a single mechanical step. It would be interesting to
study the complexity of other matrix/graph problems
assuming that the aforementioned operations can be performed
in single steps.                           

\subsection*{Acknowledgments}
The research has been
supported in part by 
Swedish Research Council grant  621-2017-03750.
\junk{
\section*{Acknowledgments}
We thank for valuable comments.

\begin{figure}\label{fig: wall}
\begin{center}
\includegraphics[height=6cm]{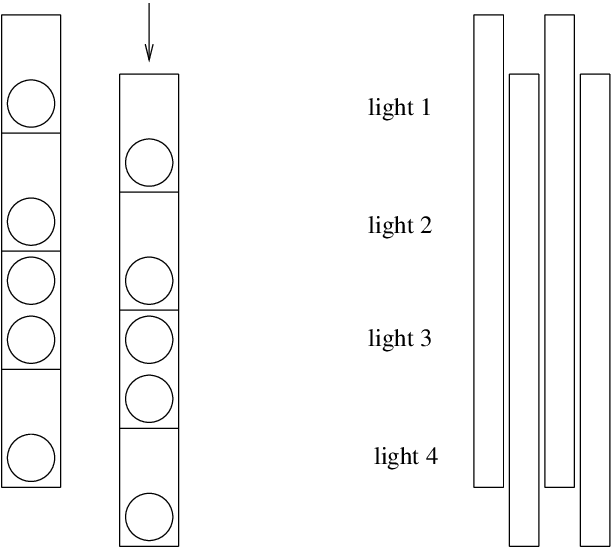}
\end{center}
\caption{(a) A thin wall modeling a not active column
of the MVP input array
with entries $1,\ 1,\ 0, \ 1.$
To activate the column the wall has to be moved down.
(b) The second and fourth columns of the MVP input array are active.}
\end{figure}}

\vfill

\end{document}